\documentclass[11pt,letterpaper,twoside,english]{article}
\usepackage[margin=1.1in]{geometry}

\usepackage{mdwlist}

\usepackage{amssymb,amsbsy,latexsym}
\usepackage{amsmath}
\usepackage{graphics, subfigure, float} 

\usepackage[latin1]{inputenc}
\usepackage[american]{babel}
\usepackage[T1]{fontenc} 
\usepackage{fourier}

\usepackage{amscd,amsthm}

\usepackage{verbatim, comment}

\usepackage[dvips]{graphicx,epsfig,color}
\usepackage{pst-all}
\usepackage{pstricks-add}
\usepackage{fp,calc}
\usepackage{datetime}
\usepackage{bm}

\newtheoremstyle{theorem}{1em}{1em}{\slshape}{0pt}{\bfseries}{.}{ }{}
\theoremstyle{theorem}
\newtheorem{theorem}{Theorem}
\newtheorem*{theorem*}{Theorem}

\newtheorem{lemma}[theorem]{Lemma}

\newtheorem*{claim*}{Claim}

\theoremstyle{remark}

\newtheorem*{remark*}{Remark}

\newtheorem*{question*}{Question}

\providecommand{\setZ}{\mathbb{Z}}

\providecommand{\setR}{\mathbb{R}}

\newcommand{\E}{\mathop{\mathbb{E}}}

\floatstyle{ruled}
\newfloat{algorithm}{tbp}{loa}
\floatname{algorithm}{Algorithm}

\newenvironment{proofofclaim}{\vspace{1ex}\noindent{\emph{Proof of claim.}}\hspace{0.5em}}
   	    {\hfill$\lozenge$\vspace{1ex}}

        \def\drawRect#1#2#3#4#5{
           \FPeval{\x2}{(#2) + (#4)} 
           \FPeval{\y2}{(#3) + (#5)} 
           \pspolygon[#1](#2,#3)(\x2,#3)(\x2,\y2)(#2,\y2)
        }

\usepackage[displaymath,textmath,sections,graphics, subfigure, floats]{preview} 
\PreviewEnvironment{center} 
\PreviewEnvironment{pspicture} 
\date{}   
\makeatother

\begin{document}
\title{A simpler proof for $O(\textrm{congestion} + \textrm{dilation})$ packet routing} 

\author{Thomas Rothvoß\thanks{Email: { \tt{rothvoss@math.mit.edu}}. Supported by the Alexander von Humboldt Foundation within the Feodor Lynen program, by ONR grant N00014-11-1-0053 and by NSF contract
CCF-0829878.} \vspace{2mm} \\ M.I.T., Cambridge, USA}

\maketitle

%
\begin{abstract}
In the \emph{store-and-forward routing} problem, packets have to be routed 
along given paths such that the arrival time of the latest packet is
minimized. 
A groundbreaking result of Leighton, Maggs and Rao says that this can always
be done in time $O(\textrm{congestion} + \textrm{dilation})$, where the
\emph{congestion} is the maximum number of paths using an edge and the \emph{dilation} is
the maximum length of a path.  
However, the analysis is quite arcane and complicated and works by iteratively
improving an infeasible schedule. 
Here, we provide a more accessible analysis which is based on 
conditional expectations. Like \cite{CongestionPlusDilation-packet-routing-LeightonMaggsRao94}, our easier analysis 
also guarantees that constant size edge buffers suffice.

Moreover, it was an open problem stated e.g. by Wiese~\cite{Dissertation-Wiese11}, whether there
is any instance where all schedules need at least $(1+\varepsilon)\cdot(\textrm{congestion}+\textrm{dilation})$ steps, 
for a constant $\varepsilon>0$. We answer this question affirmatively 
by making use of a probabilistic construction.
\end{abstract}

\section{Introduction}

One of the fundamental problems in parallel and distributed systems is to transport
packets within a communication network in a timely manner. Any routing 
protocol has to make two kinds of decisions: (1) on which paths shall the packets
be sent and (2) according to which priority rule should packets be routed along those paths, considering that communication links have usually a limited 
bandwidth. In this paper, we focus on the second part of the decision process. 
More concretely, we assume that a network in form of a directed graph  $G = (V,E)$ 
is given, together with source sink pairs $s_i,t_i \in V$ for $i=1,\ldots,k$ and $s_i$-$t_i$
paths $P_i \subseteq E$.
So the goal is to route the packets from their source along the given path
to their sink in such a way that the \emph{makespan} is minimized. 
Here, the makespan denotes the time when the last packet arrives at its destination. Moreover, we assume \emph{unit bandwidth} and \emph{unit transit time}, 
i.e. in each time unit only one packet can traverse an edge and the traversal
takes exactly one time unit.
Since the only freedom for the scheduler lies in the decision when packets move and when they wait, 
this setting is usually called \emph{store and forward routing}.
Note that we make no assumption about the structure of the graph or the paths. In fact, we can allow that the graph has multi-edges and loops; a path may even revisit
the same node several times. We only forbid that a path uses the same edge more than once. 

\begin{figure}
\begin{center}
\psset{unit=1.3cm}
\begin{pspicture}(0,0)(3,2)
\cnode*(0,0){2.5pt}{v1}
\cnode*(1,1){2.5pt}{v2}
\cnode*(2,1){2.5pt}{v3}
\cnode*(3,0){2.5pt}{v4} 
\cnode*(0,2){2.5pt}{v5}
\cnode*(3,2){2.5pt}{v6}
\ncline[arrowsize=5pt]{->}{v1}{v2}
\ncline[arrowsize=5pt]{->}{v5}{v2}
\ncline[arrowsize=5pt]{->}{v1}{v5}
\ncline[arrowsize=5pt]{->}{v2}{v3}
\ncline[arrowsize=5pt]{->}{v3}{v6}
\ncline[arrowsize=5pt]{->}{v3}{v4}
\ncline[arrowsize=5pt]{->}{v4}{v1}
\ncarc[linecolor=lightgray,arcangle=30,arrowsize=5pt,linewidth=1.5pt]{->}{v5}{v2} \naput[linecolor=lightgray,labelsep=0pt]{$\lightgray{P_1}$}
\ncarc[linecolor=lightgray,arcangle=30,arrowsize=5pt,linewidth=1.5pt]{->}{v2}{v3}
\ncarc[linecolor=lightgray,arcangle=30,arrowsize=5pt,linewidth=1.5pt]{->}{v3}{v4}
\ncarc[linecolor=lightgray,arcangle=30,arrowsize=5pt,linewidth=1.5pt]{->}{v5}{v2}

\ncarc[linecolor=gray,arcangle=-30,arrowsize=5pt,linewidth=1.5pt,linestyle=dashed]{->}{v3}{v4} \nbput[linecolor=lightgray,labelsep=0pt,npos=0.3]{$\gray{P_2}$}
\ncarc[linecolor=gray,arcangle=20,arrowsize=5pt,linewidth=1.5pt,linestyle=dashed]{->}{v4}{v1}
\ncarc[linecolor=gray,arcangle=20,arrowsize=5pt,linewidth=1.5pt,linestyle=dashed]{->}{v1}{v5}
\ncarc[linecolor=gray,arcangle=-30,arrowsize=5pt,linewidth=1.5pt,linestyle=dashed]{->}{v5}{v2}

\ncarc[linecolor=darkgray,linestyle=dotted,arcangle=-20,arrowsize=5pt,linewidth=2.5pt]{->}{v1}{v2} \nbput[linecolor=darkgray,labelsep=0pt,npos=0.7]{$\darkgray{P_3}$}
\ncarc[linecolor=darkgray,linestyle=dotted,arcangle=-20,arrowsize=5pt,linewidth=2.5pt]{->}{v2}{v3}
\ncarc[linecolor=darkgray,linestyle=dotted,arcangle=-20,arrowsize=5pt,linewidth=2.5pt]{->}{v3}{v6}
\nput{180}{v5}{$\lightgray{s_1}$}
\nput{0}{v4}{$\lightgray{t_1}$}
\nput{90}{v3}{$\gray{s_2}$}
\nput{90}{v2}{$\gray{t_2}$}
\nput{180}{v1}{$\darkgray{s_3}$}
\nput{0}{v6}{$\darkgray{t_3}$}

\end{pspicture}
\caption{Example instance with $k=3$, $C=2$ and $D=4$}
\end{center}
\end{figure}
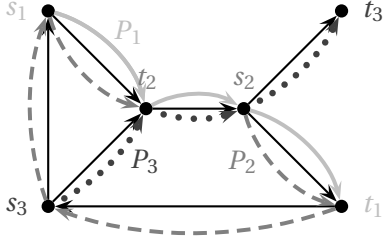

Two natural parameters of the instance are the 
 \emph{congestion} $C := \max_{e\in E} |\{ i \mid e \in P_i\}|$, i.e. the maximum number of 
paths that share a common edge and the 
\emph{dilation} $D:=\max_{i=1,\ldots,k} |P_i|$, i.e. the length of the longest path.
Obviously, for any instance, both 
parameters $C$ and $D$ are lower bounds on the makespan for
any possible routing policy. 
Surprisingly, Leighton, Maggs and Rao~\cite{CongestionPlusDilation-packet-routing-LeightonMaggsRao94} could prove that the optimum achievable makespan is always within a constant
factor of $C+D$. 
Since then, their approach has been revisited several times. 
First, \cite{Algorithms-for-Congestion-plus-Dilation-Packet-Routing-LeightonMaggsRicha99} provided a polynomial time algorithm that makes the approach constructive
(which nowadays would be easy using the Moser Tardos algorithm~\cite{ConstructiveLLL-MoserTardosJACM10}). 
Scheideler~\cite[Chapter~6]{DBLP:books/sp/Scheideler98} provides a more careful
(and more accessible) analysis which reduces the hidden constants to $39(C+D)$.
More recently Peis and Wiese~\cite{DBLP:conf/ipco/PeisW11} reduced the constant
to 24 (and beyond, for larger minimum bandwidth or transit time). 

Already the original paper of \cite{CongestionPlusDilation-packet-routing-LeightonMaggsRao94} also showed that (huge) constant size \emph{edge buffers} are sufficient.
Scheideler~\cite{DBLP:books/sp/Scheideler98} proved
that even a buffer size of 2 is enough. 
However, all proofs \cite{CongestionPlusDilation-packet-routing-LeightonMaggsRao94,Algorithms-for-Congestion-plus-Dilation-Packet-Routing-LeightonMaggsRicha99,DBLP:books/sp/Scheideler98,DBLP:conf/ipco/PeisW11} use the original idea of 
Leighton, Maggs and Rao to  start with an infeasible schedule and insert 
iteratively random delays to reduce the infeasibility until 
no more than $O(1)$ packets use an edge per time step (in each iteration, applying the Lovász Local Lemma).

In this paper, we suggest a somewhat dual approach in which we start with a 
probabilistic schedule which is feasible in expectation and then reduce step by step the 
randomness (still making use of the Local Lemma). 
Our construction here is not fundamentally different from the original 
work of \cite{CongestionPlusDilation-packet-routing-LeightonMaggsRao94}, but 
the emerging proof is ``less iterative'' and, in the opinion of the author, 
also more clear and explicit in demonstrating to the reader \emph{why} a constant 
factor suffices. Especially obtaining the additional property of 
constant size edge buffers is fairly simple in our construction.

If it comes to lower bounds for general routing strategies,  
the following instance is essentially the worst known one: $C$ many packets share the
same path of length $D$. Then it takes $C$ time units until the last packet crosses the 
first edge; that packet needs $D-1$ more time units to reach its destination, leading
to a makespan of $C+D-1$. Wiese~\cite{Dissertation-Wiese11}) states that 
no example is known where the optimum makespan needs to be even a small 
constant factor larger. 
We answer the open question in \cite{Dissertation-Wiese11} and show
that for a universal constant $\varepsilon > 0$, there is a family of instances in which every routing policy needs at 
least $(1 + \varepsilon)\cdot(C+D)$ time units (and $C,D \to \infty$)\footnote{The constant can be chosen e.g. as $\varepsilon := 0.00001$, though we do not make any attempt to optimize the constant, but focus on a simple exposition.}.
In our chosen instance, we generate paths from random permutations and 
use  probabilistic arguments for the analysis.

\subsection{Related Work}

The result of \cite{CongestionPlusDilation-packet-routing-LeightonMaggsRao94,Algorithms-for-Congestion-plus-Dilation-Packet-Routing-LeightonMaggsRicha99} could be interpreted as a constant factor approximation algorithm for the problem of 
finding the minimum makespan. In contrast, finding the optimum schedule is $\mathbf{NP}$-hard~\cite{HardnessOfStoreAndForward-ClementiDiIanniJournal96}. 
In fact, even on trees, the problem remains $\mathbf{APX}$-hard~\cite{DBLP:conf/waoa/PeisSW09}. 
If we generalize the problem to finding paths plus schedules, then
constant factor approximation algorithms are still possible due to Srinivasan and Teo~\cite{ConstantFactorPacketRouting-SrinivasanTeo00} (using the
fact that it suffices to find paths that minimize the sum of congestion and
dilation). Koch et al.~\cite{RealTimeRouting-KochPeisSkutellaWieseAPPROX09}
extend this to a more general setting,  where messages consisting of several packets have to be sent. 

The Leighton-Maggs-Rao result, apart from being 
quite involved, has the  disadvantage of being a non-local offline algorithm.
In contrast, there is a distributed algorithm with makespan $O(C) + (\log^* n)^{O(\log^* n)}D + \log^{O(1)} n$ by Rabani and Tardos~\cite{DistributedRouting-RabaniTardosSTOC96} which was later improved to $O(C + D + \log^{1+\varepsilon} n)$ 
by Ostrovsky and Rabani~\cite{DBLP:conf/stoc/OstrovskyR97}.
If the paths are indeed shortest paths, then there is a
randomized online routing policy which finishes in $O(C + D + \log k)$ steps~\cite{ShortestPathRouting-MeyerADH-Voecking99}.
To the best of our knowledge, the question concerning the existence of an $O(C+D)$ online algorithm is still open.
We refer to the book of Scheideler~\cite{DBLP:books/sp/Scheideler98} for 
a more detailed overview about routing policies.

One can also reinterpret the
packet routing problem as \emph{(acyclic) job shop scheduling}  $J \mid p_{ij}=1,\textrm{acyclic} \mid C_{\max}$, where jobs $J$ and machines 
$M$ are given. Each job has a sequence of machines that it needs to be
processed on in a given order (each machine appears at most once in this sequence), 
while all processing times have unit length. 
For the natural generalization ($J \mid p_{ij},\textrm{acyclic} \mid C_{\max}$) with arbitrary processing times $p_{ij}$, Feige \& Scheideler~\cite{AcyclicJobShops-FeigeScheideler02} showed that schedules of 
length $O(L \cdot \log L \cdot \log \log L)$ are always possible and for some instances, every schedule 
needs at least 
$\Omega(L \cdot \frac{\log L}{\log \log L})$ time units, where we abbreviate $L := \max\{C,D\}$.\footnote{In this setting, one extends $\displaystyle{C = \max_{i \in M} \sum_{j \in J: j\textrm{ uses }i} p_{ij}}$ and $\displaystyle{D = \max_{j \in J} \sum_{i \in M: j\textrm{ uses }i} p_{ij}}$.}
Svensson and Mastrolilli~\cite{FlowAndJobShopHardness-SvenssonMastrolilli-JACM11} showed that this lower bound even holds in the special case 
of \emph{flow shop scheduling}, where all jobs need to
be processed on all machines in the same order (in packet routing, this corresponds to the case
that all paths $P_i$ are identical). In fact, for \emph{flow shop scheduling with jumps} (i.e. each job needs 
to be processed on a given subset of machines) it is even $\mathbf{NP}$-hard to approximate the 
optimum makespan within any constant factor~~\cite{FlowAndJobShopHardness-SvenssonMastrolilli-JACM11}. 

In contrast, if we allow preemption, then even for acyclic job shop scheduling, the makespan can be reduced to $O(C + D\log \log \max_{ij} p_{ij})$~\cite{AcyclicJobShops-FeigeScheideler02} and it is conceivable 
that even $O(C + D)$ might suffice. 


\subsection{Organisation}

In Section~\ref{sec:Preliminaries}, we recall some probabilistic tools.
Then in Section~\ref{sec:simpleRouting} we show the existence 
of an $O(C+D)$ routing policy, which is modified in Section~\ref{sec:ConstantEdgeBuffers} to guarantee that constant size edge buffers suffice.
Finally, we show the lower bound in Section~\ref{sec:lowerBound}.

\section{Preliminaries\label{sec:Preliminaries}}

Later, we will need the following concentration result, which is a version of the
\emph{Chernov-Hoeffding bound}: 
\begin{lemma}[{\cite[Theorem~1.1]{ConcentrationOfMeasure-DubhashiPanconesi09}}] \label{lem:ChernovBound}
Let $Z_1,\ldots,Z_k \in [0,\delta]$ be independently distributed random variables with sum $Z := \sum_{i=1}^k Z_i$
and let $\mu \geq \E[Z]$.
Then for any $\varepsilon > 0$,
\[
  \Pr[Z > (1+\varepsilon)\mu] \leq \exp\Big( - \frac{\varepsilon^2 }{3} \cdot \frac{\mu}{\delta} \Big).
\]
\end{lemma}

Moreover, we need the \emph{Lovász Local Lemma} (see also the 
books \cite{ProbabilisticMethod-AlonSpencer08} and \cite{ProbAndComp-MitzenmacherUpfal05} and for the constructive version, see~\cite{ConstructiveLLL-MoserTardosJACM10}).
\begin{lemma}[Lovász Local Lemma~\cite{LovaszLocalLemma-ErdosLovasz1975}] \label{lem:LLL}
Let $A_1,\ldots,A_m$ be arbitrary events such that (1) $\Pr[A_i] \leq p$; (2) each $A_i$
depends on at most $d$ many other events; and (3) $4\cdot p\cdot d \leq 1$. Then
$\Pr\big[\bigcap_{i=1}^m \bar{A}_i\big] > 0$.
\end{lemma}

\section{$O(\textrm{congestion} + \textrm{dilation})$ routing\label{sec:simpleRouting}}

After adding dummy paths and edges, we may assume that  $C=D$ and
every path has length exactly $D$.
In the following we show how to route the 
packets within $O(D)$ time units such that in each time step, each edge is traversed
by at most $O(1)$ many packets (by stretching the time by another $O(1)$ factor, 
one can obtain a schedule with makespan $O(D)$ in which each edge is indeed only traversed
by a single packet).
In the following, we call the largest number of packets that traverse the
same edge in one time unit the \emph{load} of the schedule.

Let $\Delta > 0$ be a constant that we leave undetermined for now -- at several places
we will simply assume $\Delta$ to be large enough for our purpose.
Consider a  packet $i$ and partition its path $P_i$ into a laminar
family of \emph{blocks} such that the blocks on \emph{level $\ell$} contain
 $D_{\ell} = D^{(1/2)^{\ell}}$ many consecutive edges.\footnote{Depending on $D$, the quantity $D_{\ell}$ may not be integral. 
But all our calculations have enough slack so that one could replace $D_{\ell}$ with the nearest power of $2$. Then we may also assume that for each $\ell$, $D_{\ell}$ divides $D_{\ell-1}$. }
We stop this dissection, when the last block (whose index we denote by $L$) has length between $\Delta$ and $\Delta^2$.
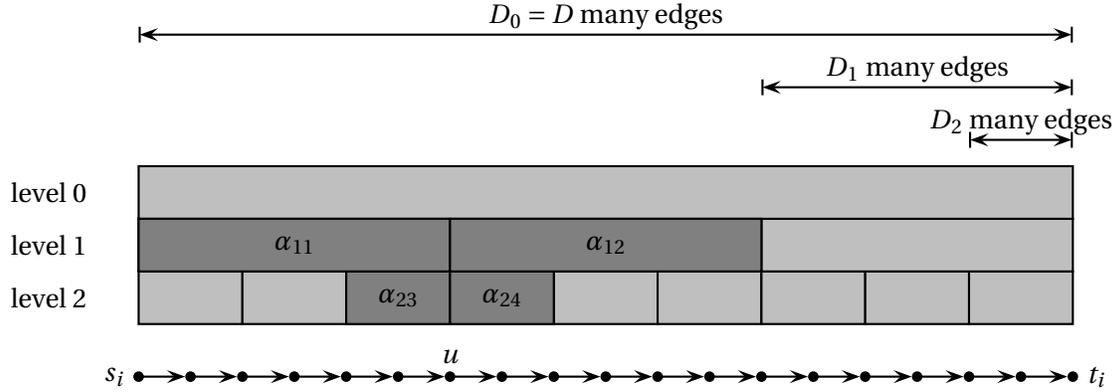
\begin{figure}
\begin{center}
\psset{xunit=0.69cm,yunit=0.7cm}
\begin{pspicture}(-1.3,0)(18,6)
\multido{\n=0+1,\N=1+1}{19}{
  \pnode(\n,0){u\n}
  \psdot[linewidth=1pt](u\n)
}
\multido{\n=0+1,\N=1+1}{18}{
  \ncline[arrowsize=5pt,nodesepB=2pt]{->}{u\n}{u\N}
}
\multido{\n=0+2}{9}{
  \drawRect{fillstyle=solid,fillcolor=lightgray}{\n}{1}{2}{1}
}
\multido{\n=0+6}{3}{
  \drawRect{fillstyle=solid,fillcolor=lightgray}{\n}{2}{6}{1}
}
  \drawRect{fillstyle=solid,fillcolor=lightgray}{0}{3}{18}{1}
  \drawRect{fillstyle=solid,fillcolor=gray}{4}{1}{2}{1}
  \drawRect{fillstyle=solid,fillcolor=gray}{6}{1}{2}{1}
  \drawRect{fillstyle=solid,fillcolor=gray}{0}{2}{6}{1}
  \drawRect{fillstyle=solid,fillcolor=gray}{6}{2}{6}{1}
  \rput[c](5,1.5){$\alpha_{23}$}
  \rput[c](7,1.5){$\alpha_{24}$}
  \rput[c](3,2.5){$\alpha_{11}$}
  \rput[c](9,2.5){$\alpha_{12}$}
 \rput[r](-1,3.5){level $0$}
 \rput[r](-1,2.5){level $1$}
 \rput[r](-1,1.5){level $2$}
 \nput{90}{u6}{$u$}
 \nput{-180}{u0}{$s_i$}
 \nput{0}{u18}{$t_i$}
 \pnode(16,4.5){D2L} \pnode(18,4.5){D2R} \ncline[arrowsize=5pt]{|<->|}{D2L}{D2R} \naput[labelsep=1pt]{$D_2$ many edges}
 \pnode(12,5.5){D1L} \pnode(18,5.5){D1R} \ncline[arrowsize=5pt]{|<->|}{D1L}{D1R} \naput[labelsep=1pt]{$D_1$ many edges}
 \pnode(0,6.5){D0L} \pnode(18,6.5){D0R} \ncline[arrowsize=5pt]{|<->|}{D0L}{D0R} \naput[labelsep=1pt]{$D_0=D$ many edges}
\end{pspicture}
\caption{Path $P_i$ and its dissection with $L = 2$. Denote the random waiting time of the $j$th block in level $\ell$ by $\alpha_{\ell j}$. Then the packet would wait for $(W_1-\alpha_{11}) + \alpha_{12} + (W_2 - \alpha_{23}) + \alpha_{24}$ time units in node $u$.\label{fig:dissection}}
\end{center}
\end{figure}

In other words, the root block (i.e. the path $P_i$ itself) 
is on level $0$ and the depth of that laminar family is $L = \Theta(\log \log D)$ 
(though this quantity will be irrelevant for the analysis).
Each block has 2 \emph{boundary nodes}, a \emph{start node} and an \emph{end node}.
Observe that a level $\ell$ block of length $D_{\ell}$ has 
children of length $D_{\ell+1} = \sqrt{D_{\ell}}$.
Moreover, we define 
\[
  W_{\ell} := \begin{cases}  D_{\ell} &  \ell = 0 \\
  D_{\ell}^{1/4} & \ell \geq 1
 \end{cases}
\]
The routing policy for packet $i$ is now as follows: 
For each level $\ell$ block, the packet waits a uniformly and independently chosen
random time $x \in [1,W_{\ell}]$ at the start node\footnote{We define $[a,b] := \{ a,a+1,a+2,\ldots,b\}$ as the set of 
integers between $a$ and $b$.}; furthermore the packet waits
 $W_{\ell}-x$ time units at the end node (see Figure~\ref{fig:dissection}). 
This policy has two crucial properties:
\begin{enumerate*}
\item[(A)] The total waiting time of each packet is $O(D)$.
\item[(B)] The time $t$ at which packet $i$ crosses an edge $e \in P_i$ is a random
variable that  depends only on the random waiting times of the blocks that
contain $e$ --- in fact, i.e. only one block from each level.
\end{enumerate*}
Let us argue, why $(A)$ is true. 
The waiting time on level $\ell=0$ will be precisely $D$, while for each
 $\ell\geq1$ the total level-$\ell$ waiting time for each packet will be $\frac{D}{D_{\ell}} \cdot W_{\ell} = \frac{D}{D_{\ell}^{3/4}}$. 
 Using the crude bound $D_{\ell} \geq 4\cdot D_{\ell+1}$ we have $D_{L-j} \geq 4^j$, 
hence  on level $L - j>0$, the total waiting time will be at most $\frac{D}{D_{L-j}^{3/4}} \leq \frac{D}{2^{j}}$.
Thus
the total waiting time for a packet, summed over all levels is at most $D + D \sum_{j=0}^{L-1} (\frac{1}{2})^j = O(D)$.
In other words: each packet is guaranteed to arrive after at most $T := O(D)$ time units.
Note that there are instances where the vast majority of random outcomes would
yield a superconstant load on some edge. However, one can prove that 
there \emph{exists} a choice of the waiting times such that the load 
does not exceed $O(1)$.

Let $X(e,t,i) \in \{ 0,1\}$ be the random variable that tells us
whether packet $i$ is crossing edge $e$ at time $t$. 
Moreover, let $X(e,t) = \sum_{i=1}^k X(e,t,i)$ be the number of packets crossing $e$
at time $t$. Since packet $i$ waits a random time from $[1,D]$ in $s_i$, we have $\Pr[X(e,t,i)] \leq \frac{1}{D}$
for each $e,i,t$ (more formally: no matter how the waiting times on level $\geq1$ are chosen, there is always at most one out of $D$ outcomes for the level $0$ waiting time that cause packet $i$ to cross $e$ precisely at time $t$). Since no edge is contained in more than $D$ paths, we have $\E[X(e,t)] \leq 1$.

In the following, if $\bm{\alpha} \in [1,W_{\ell}]^{D/D_{\ell}}$ is a vector of 
level $\ell$-waiting times, then $\E[X(e,t) \mid \bm{\alpha}]$ denotes the corresponding 
conditional expectation, depending on $\bm{\alpha}$.
The idea for the analysis is to fix the waiting times on one level at a time (starting with level $0$)
such that the conditional expectation $\E[X(e,t)]$ never increases to a value larger than, say $2$. 
Before we continue, we want to be clear about the behaviour of such conditional random 
variables.
\begin{lemma} \label{lem:ProbabilityRange}
Let $\ell \in \{ 0,\ldots,L-1\}$ and condition on arbitrary waiting times for level $0,\ldots,\ell$.
Then for any packet $i$, edge $e\in E$ and any time $t \in [T]$ one has
\begin{enumerate*}
\item[a)] $\Pr[X(e,t,i)] \leq \frac{1}{W_{\ell+1}}$.
\item[b)] If the event $X(e,t,i)$ has non-zero probability, then $\Pr[X(e,t,i)] \geq \frac{1}{W_{\ell+1}^2}$.
\end{enumerate*}
\end{lemma}
\begin{proof}
For $(a)$, suppose also all waiting times except of the level $\ell+1$ block in which $i$ 
crosses $e$ are fixed adversarially. Still, there is at most one out of $W_{\ell+1}$
outcomes that cause packet $i$ to cross $e$ at time $t$. 

For $(b)$, observe that the time at which packet $i$ crosses $e$ depends only on the waiting
time of the blocks that contain $e$ (i.e. one block per level). 
The number of possible outcomes of those waiting times
is bounded by $\prod_{j=0}^{L-\ell-1} W_{\ell+1+j} \leq (W_{\ell+1})^{\sum_{j\geq 0} (1/2)^j} = W_{\ell+1}^2$.
\end{proof}

%

The whole analysis boils down to the following lemma, in which we prove that we can always fix
the waiting times on level $\ell$ without increasing the expected load on any edge 
by more than $D_{\ell}^{-1/32}$. 
What happens formally is that we show the existence of a sequence $\bm{\alpha}_0,\ldots,\bm{\alpha}_{L-1}$ such that
$\bm{\alpha}_{\ell}$ denotes a vector of level $\ell$-waiting times and
\begin{equation} \label{eq:expectationIncrease}
   \E[X(e,t) \mid \bm{\alpha}_0,\ldots,\bm{\alpha}_{\ell-1},\bm{\alpha}_{\ell}] \leq \E[X(e,t) \mid \bm{\alpha}_0,\ldots,\bm{\alpha}_{\ell-1}]+ \frac{1}{D_{\ell}^{1/32}} \quad \forall e\in E \; \forall t\in[T]
\end{equation}
(given that the right hand side is at least $1$). To do this, suppose we already found 
and fixed proper waiting times $\bm{\alpha}_0,\ldots,\bm{\alpha}_{\ell-1}$. Then 
one can interpret the left hand side of \eqref{eq:expectationIncrease} as a random 
variable depending on $\bm{\alpha}_{\ell}$,
which is the sum of independently distributed values --- and hence well concentrated.
Moreover the dependence degree of this random variable is bounded by a polynomial in $D_{\ell}$. Thus 
the Lovász Local Lemma provides the existence of suitable waiting times $\bm{\alpha}_{\ell}$.

\begin{lemma} \label{lem:FixingOneLevel}
Let $\ell \in \{ 0,\ldots,L-1\}$ and 
suppose that we already fixed all waiting times on level $0,\ldots,\ell-1$. 
Let $X(e,t)$ be the corresponding conditional random variable
and assume $\gamma \geq \max_{e\in E,t\in[T]} \{ \E[X(e,t)] \}$ and $1\leq \gamma \leq 2$.
Then there are level $\ell$ waiting times $\bm{\alpha}$ such that
\[
\E[X(e,t) \mid \bm{\alpha}] \leq \gamma + \frac{1}{D_{\ell}^{1/32}} \quad \forall e \in E \; \forall t \in [T] 
\]
\end{lemma}
\begin{proof}
We abbreviate $m := D_{\ell}$.
First recall that on level $\ell$, (1) blocks have length $m$; (2) the child blocks have length $\sqrt{m}$ and (3) the waiting time
on the next level $\ell+1$ is from $[1,m^{1/8}]$.

We define $Y(e,t) := \E[X(e,t) \mid \bm{\alpha}]$ and consider $Y(e,t)$ as
a random variable only depending on $\bm{\alpha}$. Since the waiting times
on levels $0,\ldots,\ell-1$ are already fixed, we know exactly the level 
$\ell$-block in which packet $i$ will cross edge $e$ --- let $\alpha_{i,e}$ be the
random waiting time for that block. Then we can write 
\begin{equation} \label{eq:Y-e-t}
 Y(e,t) = \sum_{i=1}^k \Pr[X(e,t,i) \mid \alpha_{i,e}] 
\end{equation}
By Lemma~\ref{lem:ProbabilityRange}.(b), we know that $\Pr[X(e,t,i) \mid \alpha_{i,e}] \leq \frac{1}{m^{1/8}}$
for every choice of $\alpha_{i,e}$.
Thus $Y(e,t)$ is the sum of independent random variables
in the interval $[0,m^{-1/8}]$ and the Chernov bound (Lemma~\ref{lem:ChernovBound}) provides
\[
  \Pr\Big[Y(e,t) > \gamma + \frac{1}{m^{1/32}}\Big] \leq \exp\Big(- \frac{1}{3}\cdot \frac{1}{(2m^{1/32})^2} \cdot m^{1/8}\Big) \leq e^{-m^{1/16}/12}
\]
Now we want to apply the Lovász Local Lemma for the events ``$Y(e,t) > \gamma + m^{-1/32}$'' to argue that it is possible that none of the events happens. So it suffices to bound the dependence degree by 
a polynomial in $m$.
Lemma~\ref{lem:ProbabilityRange}.(b) guarantees that if the event $X(e,t,i)$ is 
possible at all, then $\Pr[X(e,t,i)] \geq \frac{1}{W_{\ell}^2} \geq \frac{1}{m}$. 
Now, reconsider Equation~\eqref{eq:Y-e-t}  and let 
$Q(e,t) := \{ i \in [k] \mid \Pr[X(e,t,i) > 0\}$ be the set of packets that still have a non-zero
chance to cross edge $e$ at time $t$. Taking expectations of Equation~\eqref{eq:Y-e-t},
we see that
\[
  2 \geq \gamma \geq \E[Y(e,t)] = \sum_{i \in Q(e,t)} \Pr[X(e,t,i)] \geq |Q(e,t)| \cdot \frac{1}{m}
\]
and hence $|Q(e,t)| \leq 2m$.
This means that each random variable $Y(e,t)$ depends on at most $2m$ entries 
of $\bm{\alpha}$. Moreover, consider an entry in $\bm{\alpha}$, say it belongs to 
packet $i$ and block $B$. This random variable appears in the definition of $Y(e,t)$ if $e \in B$ and $t$ 
belongs to $B$'s time frame -- these are just $m\cdot O(m)$ many combinations.
Here we use that the time difference between entering a level $\ell$ block and leaving it,
is bounded by $O(D_{\ell})$. 
%
Overall, the dependence degree is $O(m^3)$.
Since the probability of each bad event ``$Y(e,t) > \gamma + m^{-1/32}$'' is superpolynomially small, 
the claim follows by the Lovász Local Lemma and the assumption that $m\geq\Delta$ is large enough.
\end{proof}
We apply this lemma for $\ell=0,\ldots,L-1$ and the maximum load after any iteration will be bounded by 
$1 + \sum_{\ell=0}^{L - 1} (D_\ell)^{-1/32} \leq 2$ for $\Delta$ large enough.
The finally obtained random variables $X(e,t,i)$
are \emph{almost} deterministic --- just the waiting times on level $L$ are still 
probabilistic. But again by Lemma~\ref{lem:ProbabilityRange}, all non-zero
probabilities $\Pr[X(e,t,i)]$ are at least $\frac{1}{(\Delta^{1/4})^2} = \Omega(1)$, 
thus making an arbitrary choice for them cannot increase the load by more than 
a constant factor. 
Finally, we end up with a schedule with load  $O(1)$. 

\section{Providing constant size edge buffers\label{sec:ConstantEdgeBuffers}}

Now let us imagine that each directed edge $(u,v) \in E$ has 
an \emph{edge buffer} at the beginning of the edge. 
Whenever a packet arrives at node $u$ and has $e$ as 
next edge on its path, the packet waits in $e$'s edge buffer.
But a packet $i$ is still allowed to wait an arbitrary amount of time in $s_i$ or $t_i$.

In the construction that we saw above, it may happen that many packets wait 
for a long time in one node, i.e. a large edge buffer might be needed. However, as was shown by Leighton, Maggs and Rao~\cite{CongestionPlusDilation-packet-routing-LeightonMaggsRao94}, one 
can find a schedule such that edge buffers of size $O(1)$ 
suffice. 
More precisely, \cite{CongestionPlusDilation-packet-routing-LeightonMaggsRao94} found a schedule with load $O(1)$ in which
each packet waits at most one time unit in every node --- after stretching, 
this results in a schedule with load $1$ and $O(1)$ buffer size.

In fact, we can modify the construction from Section~\ref{sec:simpleRouting} 
in such a way that we \emph{spread} the waiting time over several edges and obtain the same property. 
Consider the dissection from the last section. Iteratively, for $\ell=1,\ldots,L$, shift 
the level $\ell$-blocks such that every level $\ell-1$ boundary node lies in the
middle of some level $\ell$-block, see Figure~\ref{fig:ShiftedDisWithRandomizationRegions} (note that we assume that $D_{\ell-1}$ is an integral multiple of $D_{\ell}$).
Fix a packet $i$ and denote the edges of its path by $P_i = (e_1,\ldots,e_D)$, then 
we assign all edges $e_j$ whose index $j$ is of the form $(1 + 2\setZ)\cdot2^{q}$
to level $L - q$ (for $q\in\{0,\ldots,L-1\}$). For example, this means that all 
odd edges are assigned to the last level; the top level does not get assigned any edges. 

Now we again define random waiting times for packet $i$ and a block $B$: on level $\ell\geq1$, each block 
picks a uniform random number $x \in [1,W_{\ell}]$. The packet waits on each of the first $x$ edges that are assigned to the block. 
Moreover, it waits on each of the last $W_{\ell}-x$ edges that are assigned to the block.
Observe that regardless of the random outcome, the packet will wait at most
once per edge since  edges are assigned to at most one level.
Using the convenient bound $2^{L-\ell} \leq D_{\ell}^{1/8}$ for $\Delta$ large enough, we see that all 
level-$\ell$ randomization takes place within the first and last $D^{3/8}_{\ell}$ edges of each block,
see Figure~\ref{fig:WaitingTimeForSingleBlock}.

\begin{figure}
\begin{center}
\psset{xunit=0.7cm,yunit=0.7cm}
\begin{pspicture}(-2,-2.5)(12,2)
\multido{\n=0+1,\N=1+1}{14}{
  \pnode(\n,0){u\n}
  \psdot[linewidth=1pt](u\n)
}
\multido{\n=0+2,\N=1+2}{7}{
  \ncline[arrowsize=5pt,nodesepB=2pt,linewidth=1.5pt]{->}{u\n}{u\N}
}
\multido{\n=1+2,\N=2+2}{6}{
  \ncline[arrowsize=5pt,nodesepB=2pt,linewidth=0.5pt,linecolor=gray,nodesepA=2pt]{->}{u\n}{u\N}
}
\psline[arrowsize=8pt]{|<-}(0,-1)(1,-1)
\psline[arrowsize=5pt,linestyle=dotted]{-}(1,-1)(2,-1)
\psline[arrowsize=5pt,linestyle=solid]{-}(2,-1)(3,-1)
\psline[arrowsize=5pt,linestyle=dotted]{-}(3,-1)(4,-1)
\psline[arrowsize=8pt]{->|}(4,-1)(5,-1)
\rput[c](2.5,-1.4){$x$ edges}
\pnode(0,-2){x1} \pnode(5,-2){x2} \ncline{|<->|}{x1}{x2} \nbput[labelsep=3pt]{$\leq D_{\ell}^{3/8}$}
\pnode(10,-2){x3} \pnode(13,-2){x4} \ncline{|<->|}{x3}{x4} \nbput[labelsep=3pt]{$\leq D_{\ell}^{3/8}$}
  \drawRect{fillstyle=solid,fillcolor=lightgray}{0}{1}{13}{1}
\pnode(0,2.5){B1} \pnode(13,2.5){B2} \ncline[arrowsize=5pt]{|<->|}{B1}{B2} \naput[labelsep=2pt]{$D_{\ell}$}

\psline[arrowsize=8pt]{|<-}(10,-1)(11,-1)
\psline[arrowsize=5pt,linestyle=dotted]{-}(11,-1)(12,-1)
\psline[arrowsize=8pt,linestyle=solid]{->|}(12,-1)(13,-1)
\rput[c](11.5,-1.4){$W_{\ell}-x$ edges}

\rput[c](0.5,0.5){$1$}
\rput[c](2.5,0.5){$1$}
\rput[c](4.5,0.5){$1$}
\rput[c](10.5,0.5){$1$}
\rput[c](12.5,0.5){$1$}
\rput[r](-1,1.5){level $\ell$:}
\rput[r](-1,0.5){waiting time:}
\end{pspicture}
\caption{Visualisation of the waiting time for a single level $\ell$-block. Black edges are assigned to the considered block and are labelled with the waiting time.\label{fig:WaitingTimeForSingleBlock}} 
\end{center}
\end{figure}
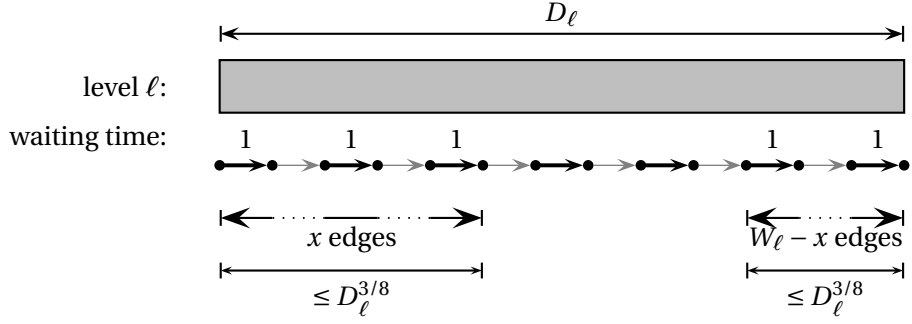

The top block does not get assigned any edge, so instead for each packet $i$, 
we pick a value $x \in [1,D]$ at random and wait $x$ time units in $s_i$.\footnote{If for a block, due to the shifting, some or all waiting edges are shifted ``before'' the source $s_i$, then the packet just waits the missing time in $s_i$.}

Reinspecting Lemma~\ref{lem:ProbabilityRange}, we observe that 
Lemma~\ref{lem:ProbabilityRange}.$b)$ holds without any alterations
and Lemma~\ref{lem:ProbabilityRange}.$a)$ holds as long as the 
considered edge $e$ has a minimum distance of $D^{3/8}_{\ell+1}$ from the 
nearest level $\ell+1$ boundary node.
Surprisingly, also Lemma~\ref{lem:FixingOneLevel} still holds with a minor modification in the claimed bound.
\begin{lemma} \label{lem:FixingOneLevelShifted}
Let $\ell \in \{ 0,\ldots,L-2\}$ and 
suppose that we already fixed all waiting times on level $0,\ldots,\ell-1$. 
Let $X(e,t)$ be the corresponding conditional random variables
and assume $\gamma \geq \max_{e\in E,t\in[T]} \{ \E[X(e,t)] \}$ and $1\leq \gamma \leq 2$.
Then there are level $\ell$ waiting times $\bm{\alpha}$ such that
\[
\E[X(e,t) \mid \bm{\alpha}] \leq \gamma + \frac{1}{D_{\ell}^{1/64}} \quad \forall e \in E \; \forall t \in [T] 
\]
\end{lemma}

\begin{proof}
Again abbreviate $m := D_{\ell}$ and consider  
\[
 Y(e,t) := \E[X(e,t) \mid \bm{\alpha}]= \sum_{i=1}^k \Pr[X(e,t,i) \mid \alpha_{i,e}] 
\]
as a random variable only depending on $\bm{\alpha}$ (recall that $\alpha_{i,e}$ is the random waiting time for that level $\ell$-block in which packet $i$ crosses edge $e$).
\begin{figure}
\begin{center}
\psset{xunit=1cm,yunit=0.7cm}
\begin{pspicture}(0,-1.9)(10,3.8)
\drawRect{fillstyle=solid,fillcolor=lightgray,linestyle=none}{0}{0}{10}{3}
\multido{\n=1.8+2.0}{4}{
  \drawRect{fillstyle=solid,fillcolor=gray,linestyle=none}{\n}{0}{0.4}{1}
}
  \drawRect{fillstyle=solid,fillcolor=gray,linestyle=none}{0}{0}{0.2}{1}
  \drawRect{fillstyle=solid,fillcolor=gray,linestyle=none}{9.8}{0}{0.2}{1}
  \drawRect{fillstyle=solid,fillcolor=gray,linestyle=none}{2.6}{1}{0.8}{1}
  \drawRect{fillstyle=solid,fillcolor=gray,linestyle=none}{6.6}{1}{0.8}{1}
  \drawRect{fillstyle=solid,fillcolor=gray,linestyle=none}{4.0}{2}{2}{1}
\multido{\n=0+2}{5}{
  \drawRect{fillstyle=none,fillcolor=lightgray,linewidth=1.0pt}{\n}{0}{2}{1}
}
\psline[linewidth=1pt](0,2)(10,2)
\psline[linewidth=1pt](0,3)(10,3)
\drawRect{fillstyle=none,fillcolor=lightgray,linestyle=solid,linewidth=1pt}{3}{1}{4}{1}
 \psline[linewidth=1pt](5,2)(5,3)
\rput[l](-2,2.5){level $\ell$}
\rput[l](-2,1.5){level $\ell+1$}
\rput[l](-2,0.5){level $\ell+2$}
\pnode(5,3.4){RL1} \pnode(6,3.4){RL2} \ncline[arrowsize=5pt]{|<->|}{RL1}{RL2} \naput[labelsep=2pt]{$\leq m^{3/8}$}
\pnode(0,4.2){L1} \pnode(5,4.2){L2} \ncline[arrowsize=5pt]{->|}{L1}{L2} \naput[labelsep=2pt]{$m$} \nput{180}{L1}{$\ldots$}
\pnode(2,-0.4){RL21} \pnode(2.2,-0.4){RL22} \ncline[arrowsize=5pt]{|-|}{RL21}{RL22} \nbput[labelsep=2pt]{$\leq m^{3/32}$}
\pnode(7,-0.4){RL11} \pnode(7.4,-0.4){RL12} \ncline[arrowsize=5pt]{|-|}{RL11}{RL12} \nbput[labelsep=2pt]{$\leq m^{3/16}$}
\pnode(0,-1.4){L21} \pnode(2,-1.4){L22} \ncline[arrowsize=5pt]{|<->|}{L21}{L22} \nbput[labelsep=2pt]{$m^{1/4}$}
\pnode(3,-1.4){L21} \pnode(7,-1.4){L22} \ncline[arrowsize=5pt]{|<->|}{L21}{L22} \nbput[labelsep=2pt]{$m^{1/2}$}
\end{pspicture}
\caption{Shifted dissection (with $m := D_{\ell}$). Regions in which randomization takes place are depicted in darkgray (observe these regions do not overlap for consecutive levels).\label{fig:ShiftedDisWithRandomizationRegions}}
\end{center}
\end{figure}
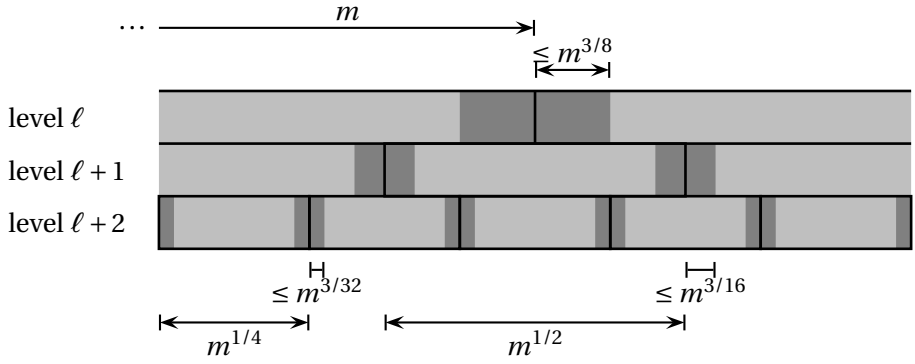
For a fixed edge $e$, for one of those levels
$\ell ' \in \{ \ell +1,\ell+2\}$, the edge $e$ is at least $\frac{1}{4}m^{1/4}$
edges away from the next level $\ell'$ boundary node (see Figure~\ref{fig:ShiftedDisWithRandomizationRegions}).
Consider the level $\ell'$-block $B$ that contains $e$.
As already argued, all randomization takes place on the
first and last $D_{\ell'}^{3/8} \leq D_{\ell+1}^{3/8} = m^{3/16} \ll \frac{1}{4}m^{1/4}$ edges (for $m\geq\Delta$ large enough).
So we can still apply Lemma~\ref{lem:ProbabilityRange}.$a)$ for level $\ell'$ 
to obtain $\Pr[X(e,t,i) \mid \alpha_{i,e}] \leq \frac{1}{W_{\ell'}} \leq \frac{1}{m^{1/16}}$. 
Again by the Chernov bound (i.e. Lemma~\ref{lem:ChernovBound} with 
$\delta := \frac{1}{m^{1/16}}$, $\varepsilon := \frac{1}{2m^{1/64}}$, $\mu := \gamma\geq1$) we have
\[
  \Pr\Big[Y(e,t) > \gamma + \frac{1}{m^{1/64}}\Big] \leq \exp\Big(- \frac{1}{3}\cdot \frac{1}{(2m^{1/64})^2} \cdot m^{1/16}\Big) = e^{-m^{1/32}/12}
\]
Next, note that still $\Pr[X(e,t,i) \mid \alpha_{i,e}] \geq \frac{1}{m}$, given
that this probability is positive. Thus from now on we can 
follow the arguments in the proof of Lemma~\ref{lem:FixingOneLevel}. The dependence degree
is still bounded by $O(m^3)$, thus the claim follows by the 
Lovász Local Lemma since $4\cdot O(m^3) \cdot e^{-m^{1/32}/12} \leq 1$ for $m\geq\Delta$ large enough.
\end{proof}
Again, we have initially $\E[X(e,t)] \leq 1$ for all $e$ and $t$, 
then we fix the waiting times iteratively on level $0,\ldots,L-2$ using Lemma~\ref{lem:FixingOneLevelShifted} and make an arbitrary choice for the waiting times
of level $L-1$ and level $L$. This results in a schedule of length $O(D)$ and load $O(1)$,
in which packets wait at most one time unit before entering an edge.

\section[A (1+epsilon)*(C + D) lower bound]{A $(1+\varepsilon)\cdot(C+D)$ lower bound\label{sec:lowerBound}}

In this section, we prove that there is an instance in which the optimum makespan 
must be at least $(1+ \varepsilon)\cdot(C + D)$, where $\varepsilon > 0$ is a small constant.
The graph $G=(V,E)$ is defined as depicted in Figure~\ref{fig:LowerBoundConstruction} (the formal definition follows from the definition of the paths, which 
we will see in a second). Edges $e_i = (u_i,v_i)$ are called \emph{critical} edges, 
while we term $(v_i,u_j)$ \emph{back edges}.
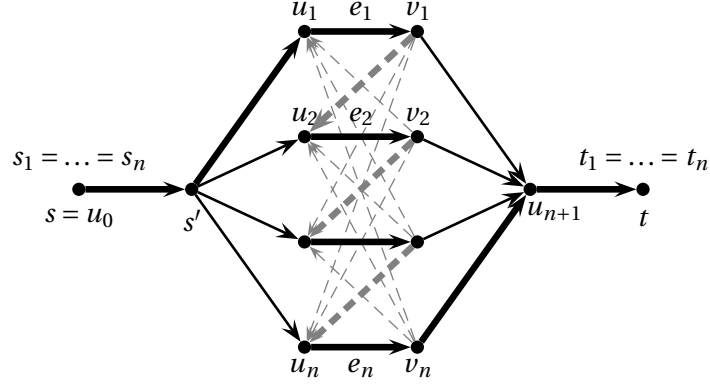
\begin{figure} 
\begin{center}
\psset{xunit=1.5cm,yunit=1.4cm}
\begin{pspicture}(0,-1.8)(5,1.8)
\cnode*(0,0){2.5pt}{s}  \nput{-90}{s}{$s=u_0$} \nput{90}{s}{$s_1=\ldots=s_n$}
\cnode*(1,0){2.5pt}{s2} \nput{-90}{s2}{$s'$}
\cnode*(2,1.5){2.5pt}{u1}
\cnode*(2,0.5){2.5pt}{u2}
\cnode*(2,-0.5){2.5pt}{u3}
\cnode*(2,-1.5){2.5pt}{u4}
\cnode*(3,1.5){2.5pt}{v1}
\cnode*(3,0.5){2.5pt}{v2}
\cnode*(3,-0.5){2.5pt}{v3}
\cnode*(3,-1.5){2.5pt}{v4}
\cnode*(4,0){2.5pt}{t2}
\cnode*(5,0){2.5pt}{t} \nput{90}{t}{$t_1=\ldots=t_n$}\nput{-90}{t}{$t$}
\nput[labelsep=2pt]{90}{u1}{$u_1$}
\nput[labelsep=2pt]{90}{u2}{$u_2$}
\nput[labelsep=2pt]{-90}{u4}{$u_n$}
\nput[labelsep=2pt]{90}{v1}{$v_1$}
\nput[labelsep=2pt]{90}{v2}{$v_2$}
\nput[labelsep=2pt]{-90}{v4}{$v_n$}
\nput[labelsep=3pt]{-60}{t2}{$u_{n+1}$}
\ncline[arrowsize=7pt,linewidth=2.5pt,linestyle=dashed,linecolor=gray]{->}{v1}{u2}
\ncline[arrowsize=4pt,linewidth=0.5pt,linestyle=dashed,linecolor=gray]{->}{v1}{u3}
\ncline[arrowsize=4pt,linewidth=0.5pt,linestyle=dashed,linecolor=gray]{->}{v1}{u4}
\ncline[arrowsize=4pt,linewidth=0.5pt,linestyle=dashed,linecolor=gray]{->}{v2}{u1}
\ncline[arrowsize=4pt,linewidth=2.5pt,linestyle=dashed,linecolor=gray]{->}{v2}{u3}
\ncline[arrowsize=4pt,linewidth=0.5pt,linestyle=dashed,linecolor=gray]{->}{v2}{u4}
\ncline[arrowsize=4pt,linewidth=0.5pt,linestyle=dashed,linecolor=gray]{->}{v3}{u1}
\ncline[arrowsize=4pt,linewidth=0.5pt,linestyle=dashed,linecolor=gray]{->}{v3}{u2}
\ncline[arrowsize=4pt,linewidth=2.5pt,linestyle=dashed,linecolor=gray]{->}{v3}{u4}
\ncline[arrowsize=4pt,linewidth=0.5pt,linestyle=dashed,linecolor=gray]{->}{v4}{u1}
\ncline[arrowsize=4pt,linewidth=0.5pt,linestyle=dashed,linecolor=gray]{->}{v4}{u2}
\ncline[arrowsize=4pt,linewidth=0.5pt,linestyle=dashed,linecolor=gray]{->}{v4}{u3}

\ncline[arrowsize=6pt,linewidth=2.5pt]{->}{s}{s2}
\ncline[arrowsize=6pt,linewidth=2.5pt]{->}{s2}{u1}
\ncline[arrowsize=6pt,linewidth=1pt]{->}{s2}{u2}
\ncline[arrowsize=6pt,linewidth=1pt]{->}{s2}{u3}
\ncline[arrowsize=6pt,linewidth=1pt]{->}{s2}{u4}
\ncline[arrowsize=6pt,linewidth=2.5pt]{->}{u1}{v1} \naput[labelsep=2pt]{$e_1$}
\ncline[arrowsize=6pt,linewidth=2.5pt]{->}{u2}{v2} \naput[labelsep=2pt]{$e_2$}
\ncline[arrowsize=6pt,linewidth=2.5pt]{->}{u3}{v3} 
\ncline[arrowsize=6pt,linewidth=2.5pt]{->}{u4}{v4} \nbput[labelsep=2pt]{$e_n$}
\ncline[arrowsize=6pt,linewidth=1pt]{->}{v1}{t2} 
\ncline[arrowsize=6pt,linewidth=1pt]{->}{v2}{t2} 
\ncline[arrowsize=6pt,linewidth=1pt]{->}{v3}{t2} 
\ncline[arrowsize=6pt,linewidth=2.5pt]{->}{v4}{t2} 
\ncline[arrowsize=6pt,linewidth=2.5pt]{->}{t2}{t}

\end{pspicture}
\caption{Graph $G=(V,E)$. The bold edges depict the path $P_i$ originating from the identity permutation $\pi_i=(1,2,\ldots,n)$.\label{fig:LowerBoundConstruction}}
\end{center}
\end{figure}
We want to choose paths $P_1,\ldots,P_n$ as random paths though the network, 
all starting at $s_i:=s$ and ending at $t_i:=t$.
More concretely, each packet $i$ picks a uniform random permutation $\pi_i : [n] \to [n]$
which gives the order in which it moves through the critical edges $e_1,\ldots,e_n$.  
In other words, 
\[
P_i = (s,s',u_{\pi_i(1)},v_{\pi_i(1)},u_{\pi_i(2)},v_{\pi_i(2)},\ldots,,u_{\pi_i(n)},v_{\pi_i(n)},u_{n+1},t).
\]
Then the congestion is $n$ and the dilation is $2n+3$.
We consider the time frame $[1,T]$ with $T = (3+\varepsilon)n$ and claim that for $\varepsilon>0$ small 
enough, there will be no valid routing that is finished by time $T$. 
\begin{theorem} \label{thm:lowerBound}
Pick paths $P_1,\ldots,P_n$ at random. Then with probability $1-e^{-\Omega(n^2)}$, there
is no packet routing policy with makespan at most $3.000032n$
(even if buffers of unlimited size are used).
\end{theorem}
First of all, clearly the makespan must be at least $C+D-1 \approx 3n$
since all paths have the same length $D$ and all packets must first cross
edge $(s,s')$. So if we allow only time $(3+\varepsilon)n$, then there is only 
a small slack of $\varepsilon n$ time units. One can show that the number of 
different possible routing strategies is bounded by $2^{o(n^2)}$ (for $\varepsilon\to0$).
In contrast, we can argue that a \emph{fixed} routing will fail
against random paths with probability $2^{-\Omega(n^2)}$. Then choosing
$\varepsilon$ small enough, the theorem follows using the union bound over
all routing strategies.

We call a packet $i$ \emph{active} at time $\tau$ if it is traversing an edge. 
We say a packet is \emph{parking} at time $\tau$ if it is either in the end node $t_i$ nor
in the start node $s_i$. We say a packet is \emph{waiting} if it is neither active nor parking.

\subsection{The number of potential routing strategies}

Consider a fixed packet $i$ and let us discuss, how a routing strategy 
is defined. The only decision that is made, is of the form: ``\emph{How many time
units shall the packet wait in the $k$-th node on its path (for $k=0,\ldots,D$)''}.
It is not necessary to wait in $s'$ since a packet could instead move to
$u_{\pi_i(1)}$ and wait there. Moreover, it is not needed to wait in one of the nodes $v_j$, 
since instead it could also wait in the next $u_{j'}$ node on its way (the reason is
that if there would be a collision on a back edge $(v_{\pi_i(j)},u_{\pi_i(j+1)})$ with packet $i'\neq i$, 
then this packet $i'$ has crossed the critical edge $(u_{\pi_i(j)},v_{\pi_i(j)})$ together with $i$ 
in the previous time step, so there was already a collision). 
In other words, the complete routing strategy for packet $i$ can be described
as a $(n+2)$-dimensional vector $W_i \in \setZ_{\geq0}^{n+2}$, where
$W_{ij}$ is the time that packet $i$ stays in node $u_{j}$ (for convenience, we denote $s$ also as $u_0$).
Then  $\sum_{j=1}^{n+1} W_{ij}$ is the total waiting time and for $i\in[n]$ and
$W_{i0}$ is the time that $i$ parks in the start node.  

Independently from the outcome of the random experiment, we know the time
when each packet crosses the edges incident to $s$ and to $t$. 
We call $W$ a \emph{candidate routing strategy}, if there is no collision
on $(s,s')$ and $(u_{n+1},t)$ and the makespan of each packet is bounded by
$(3+\varepsilon)n$.

Recall that $H(\delta)=\delta \log\frac{1}{\delta} + (1-\delta) \log \frac{1}{1-\delta}$ is the \emph{binary entropy function}\footnote{Here $\log$ is the binary logarithm.}. Then we have:
\begin{lemma}
The total number of candidate routing matrices $W$ is at most $2^{(\Phi(\varepsilon)+o(1))\cdot n^2}$, 
where  $\Phi(\varepsilon) := H(\frac{\varepsilon}{1+\varepsilon})\cdot (1+\varepsilon)$. 
\end{lemma}
\begin{proof}
First of all, the parking times in $s$ and the total waiting time $\sum_{j=1}^{n+1} W_{ij}$ 
for a packet $i$ are between $0$ and $(1+\varepsilon)n \leq 2n$; thus there are at most 
 $(2n)^{2n} = 2^{o(n^2)}$ many possibilities to choose them.

Thus assume that the total waiting time $\varepsilon_in = \sum_{j=1}^{n+1} W_{ij}$ for packet $i$ is fixed.
Then the number of possibilities how this waiting time can be distributed among nodes $u_1,\ldots,u_{n+1}$ 
is bounded by 
\[
 {(n+1) + (\varepsilon_i n) - 1 \choose \varepsilon_i n} \leq 2^{H(\frac{\varepsilon_i}{1+\varepsilon_i})\cdot (1+\varepsilon_i)\cdot n} = 2^{\Phi(\varepsilon_i)\cdot n}
\]
where we use the bound ${m \choose \delta m} \leq 2^{H(\delta)m}$ with $m = (1+\varepsilon_i)n$ and $\delta = \frac{\varepsilon_i}{1+\varepsilon_i}$.
%
%

Next, let us  upperbound the total waiting time $n\sum_{i=1}^n \varepsilon_i$. Of course, the waiting time must fit
into the time frame of length $T=(3+\varepsilon)n$. Since edge $(s,s')$ can only be crossed by one packet at
a time, the cumulated time that the packets spend in the start node is
at least $\sum_{\tau=0}^{n-1} \tau \approx \frac{n^2(1-o(1))}{2}$. The same amount of time is spent by all packets in the end node. 
Moreover, the packets spend at least $2n^2$ time units traversing edges. We conclude that
\[
 n\sum_{i=1}^n \varepsilon_i \leq nT - \frac{n^2(1-o(1))}{2} - \frac{n^2(1-o(1))}{2} - 2n^2 = (\varepsilon+o(1)) n^2,
\]
thus $\sum_{i=1}^n \varepsilon_i \leq (\varepsilon+o(1))n$.
Once the values $\varepsilon_1,\ldots,\varepsilon_n$ are fixed, the total number of routing policies for the $n$ packets is hence upperbounded by 
\[
 \prod_{i=1}^n 2^{\Phi(\varepsilon_i) n}
= 2^{n\sum_{i=1}^n \Phi(\varepsilon_i) }
\leq 2^{n^2 \Phi(\frac{1}{n}\sum_{i=1}^n \varepsilon_i)}
\leq 2^{n^2 (\Phi(\varepsilon) + o(1))}
\]
Here we use Jensen's inequality together with the fact that $\Phi$ is concave. 
The claim follows.
\end{proof}
The important property of function $\Phi$ apart from 
concavity is that $\lim_{\varepsilon \to 0} \Phi(\varepsilon) = 0$. 
Note that for $0\leq\varepsilon\leq\frac{1}{10}$, one can conveniently upperbound $\Phi(\varepsilon) \leq 2^{1.5\varepsilon\log(\frac{1}{\varepsilon})n}$.
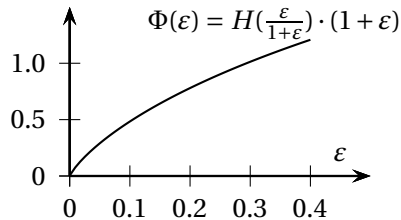
\begin{figure}[H]
\begin{center}
\psset{yunit=1.5cm,xunit=8cm}
\begin{pspicture}(0,-0.2)(0.5,1.6)
\psaxes[Dx=0.1,Dy=0.5,arrowsize=5pt]{->}(0,0)(0,0)(0.5,1.5)
\psplot[algebraic=true]{0.001}{0.4}{(1+x)*( (x/(1+x)) * log((1+x)/x)/log(2) + (1-x/(1+x))*log(1/(1-x/(1+x)))/log(2)  )}
\rput[l](0.13,1.35){$\Phi(\varepsilon) = H(\frac{\varepsilon}{1+\varepsilon})\cdot (1+\varepsilon)$}
\rput[c](0.45,8pt){$\varepsilon$}
\end{pspicture}
\end{center}
\caption{Function $\Phi$.\label{fig:FunctionPhi}}
\end{figure}


\subsection{A fixed strategy vs. random paths}

Now consider a \emph{fixed} candidate routing matrix $W$ and imagine that the paths are taken at random. 
We will show that this particular routing matrix $W$ is not legal with probability $1-e^{\Omega(n^2)}$.
For this sake, we observe that there must be $\Omega(n)$ time units in which at least a constant
fraction of packets cross critical edges. For each such time unit the probability of
having no collision is at most $(\frac{1}{2})^{\Omega(n)}$ and the claim follows.
The only technical difficulty lies in the fact that the outcomes of values 
$\pi_i(j)$ and $\pi_i(j')$ for the random permutations are (mildly) dependent.
\begin{lemma}
Suppose $\varepsilon \leq \frac{1}{20}$. Let $W$ be a candidate routing matrix. Then take paths  $P_1,\ldots,P_n$ at random.
The probability that the routing scheme defined by $W$ is collision-free is 
at most $(\frac{15}{16})^{n^2/128}$.
\end{lemma}
\begin{proof}
For time $\tau$, let $\beta_{\tau}n$ be the number of packets that cross one of the
critical edges at time $\tau$, thus $\sum_{\tau=1}^T \beta_{\tau} = n$ (note that the $\beta_{\tau}$'s do not depend on the random experiment).
Let $p := \Pr_{\tau\in[T]}[\beta_{\tau} \geq \frac{1}{4}]$ be the fraction of time units in which at least $\frac{n}{4}$ packets are
crossing a critical edge. Then
\[
 \frac{1}{3+\varepsilon} = \frac{\sum_{\tau=1}^T \beta_{\tau}}{T} = \E_{\tau \in[T]}[\beta_{\tau}] \leq 1\cdot p + (1-p)\cdot\frac{1}{4},
\]
which can be rearranged to $p \geq \frac{1}{10}$ for $\varepsilon \leq \frac{1}{20}$.
In other words, we have $\frac{T}{10} \geq \frac{1}{16}n =: k$ many time units $\bm{\tau}=\{\tau_1,\ldots,\tau_{k}\}$ in which at least $\frac{n}{4}$
many packets are crossing an edge in $e_1,\ldots,e_n$.
Let $A(\tau)$ be the event that there is no collision at time $\tau$. 
Then we can bound the probability of having no collision at all, by just considering
the time units in $\bm{\tau}$:
\[
  \Pr\Big[\bigwedge_{\tau=1}^T A(\tau)\Big] \leq \prod_{j=1}^{k} \Pr[A(\tau_j) \mid A(\tau_1),\ldots,A(\tau_{j-1})]
\stackrel{(*)}{\leq} \left(\frac{15}{16}\right)^{\frac{n}{8} \cdot k} = \left(\frac{15}{16}\right)^{n^2/128}
\]
It remains to justify the inequality $(*)$.
\begin{claim*}
For all $j=1,\ldots,k$ one has $\Pr[A(\tau_j) \mid A(\tau_1),\ldots,A(\tau_{j-1})] \leq (\frac{15}{16})^{n/8}$.
\end{claim*}
By $P_i(\tau)$ we denote the random variable that gives the edge that $i$ traverses at time $\tau$
(in case that $i$ is waiting in a node $v$, let's say that $P_i(\tau) = (v,v)$). 
Let $E_i := \{ P_i(\tau_{1}),\ldots,P_i(\tau_{j-1}) \} \cap \{ e_1,\ldots,e_n\}$ be the 
critical edges that packet $i$ has visited at $\tau_1,\ldots,\tau_{j-1}$. 
It suffices to show that $\Pr[A(\tau_j) \mid E_1,\ldots,E_n] \leq (\frac{15}{16})^{n/16}$, i.e. 
we condition on those edges $E_i$.
Let $I \subseteq [n]$ with $|I| = \frac{n}{4}$ be the indices of packets that cross a critical edge 
at time $\tau_j$. We split $I$ into equally sized parts $I = I_1 \dot{\cup} I_2$, i.e. $|I_1| = |I_2| = \frac{n}{8}$.
Consider the critical edges $E^* := \{ P_i(\tau_j) \mid i \in I_1\}$ which 
are chosen by packets in $I_1$. 
If $|E^*| < \frac{n}{8}$ then we have a collision, so condition on the event that $|E^*| = \frac{n}{8}$.
Now for all other packets $i \in I_2$, the edge $P_i(\tau_j)$ is a uniform random choice 
from $\{e_1,\ldots,e_n\} \backslash E_i$. Thus we have independently for all $i \in I_2$, 
\[
  \Pr[P_i(\tau_j) \in E^*] = \frac{|E^* \backslash E_i|}{|\{e_1,\ldots,e_n\} \backslash E_i|} \geq \frac{n/8 - n/16}{n} = \frac{1}{16},
\]
since $|E_i| \leq k=\frac{n}{16}$.
Thus
\[
 \Pr\Big[A(\tau_j) \mid \; |E^*| = \frac{n}{8}; \; E_1,\ldots,E_n\Big] \leq \Pr\Big[\bigwedge_{i\in I_2} P_i(\tau_j) \notin E^* \mid \; |E^*| = \frac{n}{8}; \; E_1,\ldots,E_n\Big] \leq \left(\frac{15}{16}\right)^{n/8}
\]
and the claim follows. 
\end{proof}

Finally one can check that for $\varepsilon := 0.000032$ and $n$ large enough one has
\[
 \left(\frac{15}{16}\right)^{n^2/128} \cdot 2^{(\Phi(\varepsilon) + o(1))n^2} < 1
\]
and Theorem~\ref{thm:lowerBound} follows.

Observe that in our instance, $C$ and $D$ are within a factor of $2$ or each other.
In contrast, if $C \gg D$, then there is a schedule of length $(1+o(1))\cdot C$
and buffer size $O(\frac{C}{D})$, see~\cite[Chapter 6]{DBLP:books/sp/Scheideler98}.

\paragraph{Acknowledgements.}

The author is very grateful to Rico Zenklusen for carefully reading
a preliminary draft.

\bibliographystyle{alpha}
\bibliography{simplepacketrouting}

\end{document}